\documentclass[12pt]{article} \usepackage{fullpage}
\usepackage{amssymb,amsmath}
\usepackage{wasysym}
\usepackage{lineno}

\usepackage{nicefrac}
\usepackage{url}
\usepackage{color}
\usepackage[american]{babel}
\usepackage{graphicx}
\usepackage{version}
\usepackage{subfigure} 
\usepackage[textsize=footnotesize]{todonotes}

\usepackage{import}

\newif\iffull
\fulltrue
\newif\ifoldmiddle
\oldmiddlefalse
\newif\ifoldhard
\oldhardfalse

\graphicspath{{figures/}}

\usepackage{amsthm}
\newtheorem{theorem}{Theorem}
\newtheorem{corollary}{Corollary}

\newtheorem{claim}{Claim}
\newtheorem{observation}{Observation}

\newtheorem{lemma}{Lemma}

\newcommand{\leaveout}[1]{}

\date{}
\title{Finding large matchings in 1-planar graphs of minimum degree 3}
\author{Therese Biedl
\thanks{David R.~Cheriton School of Computer
Science, University of Waterloo, Waterloo, Ontario N2L 1A2, Canada.
Supported by NSERC. {\it biedl@uwaterloo.ca}} 
\and{Fabian Klute}
\thanks{Algorithms and Complexity Group, TU Wien, Austria, {\it fklute@ac.tuwien.ac.at}.  Research initiated
while the author was visiting the University of Waterloo.}
}

\begin{document}

\maketitle

\begin{abstract}
A {\em matching} is a set of edges without common endpoint.  It was recently shown that every {\em 1-planar graph} (i.e., a graph that can be drawn in the plane with at most one crossing per edge) that has minimum degree 3 has a matching of size at least $\frac{n+12}{7}$, and this is tight for some graphs.  The proof did not come with an algorithm to find the matching more efficiently than a general-purpose maximum-matching algorithm.  In this paper, we give such an algorithm.  More generally, we show that any matching that has no augmenting paths of length 9 or less has size at least $\frac{n+12}{7}$ in a 1-planar graph with minimum degree 3.
\end{abstract}

\section{Introduction}
The {\em matching problem} (i.e., finding a large set of edges in a graph 
such that no two chosen edges have a common endpoint) is one of the
oldest problem in graph theory and graph algorithms, see for example
\cite{Berge,LP86} for overviews.  

To find the maximum matching in a graph $G=(V,E)$, the fastest algorithm is
the one by Hopcroft and Karp if $G$ is bipartite \cite{HK73},
and the one by Micali and Vazirani otherwise (\cite{micaliSqrtAlgorithm1980}, see also
\cite{Vaz12} for further clarifications).  As pointed out in \cite{Vaz12}, 
for a graph with $n$ vertices and $m$ edges
the run-time of the algorithm by Micali and Vazirani is $O(m\sqrt{n})$ in the
RAM model and $O(m\sqrt{n} \alpha(m,n))$ in the pointer model,
where $\alpha(\cdot)$ is the inverse Ackerman function.  For {\em planar graphs} (graphs that
can be drawn without crossing in the plane) there exists a linear-time
approximation scheme for maximum matching \cite{Baker94}, and it can
easily be generalized to so-called $H$-minor-free graphs \cite{DFHT05}
and $k$-planar graphs \cite{GB07}.

For many graph classes, specialized results concerning matchings and
matching algorithms have been found.  To name just a few, every
bipartite $d$-regular graph has a {\em perfect matching} (a
matching of size $n/2$) \cite{Hal56} and it can be found in $O(m\log d)$ 
time \cite{COS01}.  Every 3-regular biconnected graph has a perfect
matching \cite{Petersen} and it can be found in linear time for planar
graphs and in near-linear time for arbitrary graphs \cite{BBD+01}.  Every
graph with a Hamiltonian path has a {\em near-perfect matching} (of
size $\lceil (n{-}1)/2\rceil$); this includes for example the 4-connected
planar graphs \cite{Tutte} for which the Hamiltonian path (and with it
the near-perfect matching) can be found in linear time \cite{CN89}.

For graphs that do not have perfect or near-perfect matchings, one
possible avenue of exploration is to ask for guarantees on the size
of matchings.  One of the first results in this direction is due to Nishizeki and
Baybars \cite{NB79}, who showed that every planar graph with 
minimum degree~3 has a matching of size at least $\frac{n+4}{3}$.
(This bound is tight for some planar graphs with minimum degree 3.)
The proof relies on the Tutte-Berge theorem and does not give an
algorithm to find such a matching (or at least, none faster than any
maximum-matching algorithm).  Over 30 years later, a linear-time 
algorithm to find a matching of this size in planar graphs of minimum degree 3
was finally developed by Franke, Rutter and Wagner \cite{frankeComputingLarge2011}.
The latter paper was a major inspiration for our current work.

In recent years, there has been much interest in {\em near-planar}
graphs, i.e., graphs that may be required to have crossings but
that are ``close'' to planar graphs in some sense.  We are interested
here in {\em 1-planar graphs}, which are those that can be drawn with
at most one crossing per edge.  (Detailed definitions can be found in Section~\ref{sec:background}.)   
See a recent annotated bibliography \cite{KLM17} for
an overview of many results known for 1-planar graphs.
The first author and Wittnebel \cite{BW19} gave matching-bounds for
1-planar graphs of varying minimum degrees, and showed that any 1-planar graph with
minimum degree 3 has a matching of size at least $\frac{n+12}{7}$.  (Again, this
bound is tight for some 1-planar graphs with minimum degree 3.)  

The proof in \cite{BW19} is again via the Tutte-Berge theorem and
does not give rise to a fast algorithm to find a matching of this size.
This is the topic of the current paper.  We give an algorithm that
finds, for any 1-planar graph with minimum degree 3, a matching of
size at least $\frac{n+12}{7}$ in linear time in the RAM model
and time $O(n\alpha(n))$ in the pointer-model.  The algorithm consists
simply of running the algorithm by Micali and Vazirani for a limited
number of rounds (and in particular, does not require that a 1-planar
drawing of the graph is given).  The bulk of the work consists of the
analysis, which states that if there are no augmenting paths of length
9 or less, then the matching has the desired size for graphs with
minimum degree 3.  Along the way, we also prove some bounds obtained
for graphs with higher minimum degree, though these are not tight.

The paper is structured as follows.  After reviewing some background
in Section~\ref{sec:background}, we state the algorithm in 
Section~\ref{sec:algo}.  The analysis proceeds in multiple steps
in Section~\ref{sec:analysis}.
We first delete short flowers from the graph (and account for free
vertices in them directly).  The remaining graph is basically bipartite,
and we can use bounds known for independent sets in 1-planar graphs to
obtain matching-bounds that are very close to the desired goal.
Closing this gap requires a non-trivial modification of the
graph and argument; this is deferred to Section~\ref{sec:appendix_hard}
before we conclude in Section~\ref{sec:conclusion}.

\section{Background}
\label{sec:background}

We assume familiarity with graphs and graph algorithms, see for example
\cite{Die12} and \cite{SchaeferBook}.  Throughout the paper, $G$ is a 
simple graph with $n$ vertices and $m$ edges.  A {\em matching} is
a set $M$ of edges without common endpoints; we say that $e=(x,y)\in M$
is {\em matched} and $x$ and $y$ are {\em matching-partners}.  
$V(M)$ denotes the endpoints of edges in $M$; we call $v\in V(M)$
{\em matched} and all other vertices {\em free}.  An {\em 
alternating walk} is a walk that alternates
between unmatched and matched edges.
An {\em augmenting path} is an
alternating walk that repeats no vertices and begins and ends at a free vertex;
we use {\em $k$-augmenting path} for an augmenting path with at most $k$ edges.
Note that if $P$ is an augmenting
path of $M$ (and viewed as an edge-set), then $(M\setminus P) \cup (P\setminus M)$
is also a matching and has one more edge.  

A {\em drawing} $\Gamma$ of a graph consists of assigning points in $\mathbb{R}^2$
to vertices and simple curves to each edge such that curves of edges end at the
points of its endpoints.  We usually identify the graph-theoretic object (vertex,
edge) with the geometric object (point, curve) that it has been assigned to.
A drawing is called {\em good} (see \cite{SchaeferBook} for details)
if (1) no edge intersects a point of a non-incident vertex, (2) two edges intersect in at most
We only consider {\em good} drawings (see \cite{SchaeferBook} for details)
that avoid degeneracies such an edge going through the point of a non-incident
vertex or two edges with a common endpoint intersecting.
The connected regions of $\mathbb{R}^2\setminus \Gamma$
are called the {\em regions} of the drawing.

A {\em crossing} $c$ of $\Gamma$ is a pair of two edges $(v,w)$ and $(x,y)$
that have a point in their interior in common.
A {\em crossed} edge is one that has a crossing on it; otherwise it is called
{\em uncrossed}.
A drawing $\Gamma$ is called {\em $k$-planar} (or {\em planar} for $k=0$) if every edge has at most $k$ crossings.
A graph is called {\em $k$-planar} if it has a $k$-planar drawing.  While planarity
can be tested in linear time \cite{HT74,BL76}, testing 1-planarity
is NP-hard \cite{GB07}.

Fix a 1-planar drawing  $\Gamma$ and consider a crossing $c$ between edges $(v_0,v_2)$ and $(v_1,v_3)$.
Then we could draw edge $(v_i,v_{i+1})$ (for $i=0,\dots,3$ and addition modulo 4) without
crossing by walking ``very close'' to crossing $c$.  We call the pair $(v_i,v_{i+1})$ a {\em potential
kite-edge} and note that if we inserted $(v_i,v_{i+1})$ in the aforementioned manner, then it would be consecutive
with the crossing edges in the cyclic orders of edges around $v_i$ and $v_{i+1}$ in $\Gamma$.

\section{Finding the matching}
\label{sec:algo}

\ifoldmiddle
\def\maxlength{7}
\else
\def\maxlength{9}
\fi

Our algorithm to find a large matching is a one-liner: repeatedly
extend the matching via \maxlength-augmenting paths (i.e., of
length at most \maxlength) until there are no more such paths.

Note that the algorithm does not depend on the knowledge that the
graph is 1-planar and does not require having a 1-planar drawing
at hand.  It could be executed on any graph;
our contribution is to show (in the next section) that if 
it is executed on a 1-planar graph $G$ with minimum degree 3 then
the resulting matching has size at least $\frac{n+12}{7}$.

\paragraph{Running time}
Finding a matching $ M $ in $ G $ such that 
there is no $k$-augmenting path can be done in time $ O(k|E|) $ 
using the algorithm by Micali and Vazirani~\cite{micaliSqrtAlgorithm1980}.
(We state all run-time bounds here in the RAM model; for the pointer
model add a factor of $\alpha(|E|,|V|)$.)
This algorithm runs in phases, each of which has a running time of $ O(|E|) $ and 
increases the length of the minimum-length augmenting path by at least two.
See for example the paper by Bast et al.~\cite{bastMatchingAlgorithms2006} for a more detailed explanation.
Since for 1-planar graphs we have $ |E| \in O(|V|) $ we get a linear time algorithm in the number of vertices of $ G $
to find a matching without \maxlength-augmenting paths. 

\section{Analysis}
\label{sec:analysis}

Assume that $M$ is a matching without augmenting paths of length at most 9,
and let $F$ be the free vertices; $|F|=n-2|M|$.
To analyze the size of $M$, we proceed in three stages.  First we remove some
vertices and matching-edges that belong to short flowers (defined below); these
are ``easy'' to account for.  Next we split the remaining vertices by their
distance (measured along alternating paths) to free vertices.  Since short
flowers have been removed, no edges can exist between vertices of even
small distance; they hence form an independent set.
Using a crucial
lemma from \cite{BW19} on the size of independent sets in 1-planar graphs, 
this shows that $|M|\geq \frac{7}{50}(n+12)$,
which is very close to the desired bound of $\frac{n+12}{7}$.  
The last stage (which does the improvement from $\frac{7}{50}$ to $\frac{1}{7}$) will require 
non-trivial effort and is done mostly out of 
academic interest; this is deferred to Section~\ref{sec:appendix_hard}.

\paragraph{Flowers}

A {\em flower}%
\footnote{Our terminology follows the one in Edmonds'
famous blossom-algorithm \cite{Edm65}.}
 is an alternating walk that begins and ends at 
the same free vertex; we write {\em $k$-flower} for a flower with at
most $k$ edges. We only consider 7-flowers;
Figure~\ref{fig:flowers} illustrates all possible such flowers.
Note that such short flowers split into a path (called {\em stem})
and an odd cycle (the {\em blossom}); we call a flower a
{\em cycle-flower} if the stem is empty.  

\begin{figure}[ht]
\hspace*{\fill}
	\includegraphics[page=1]{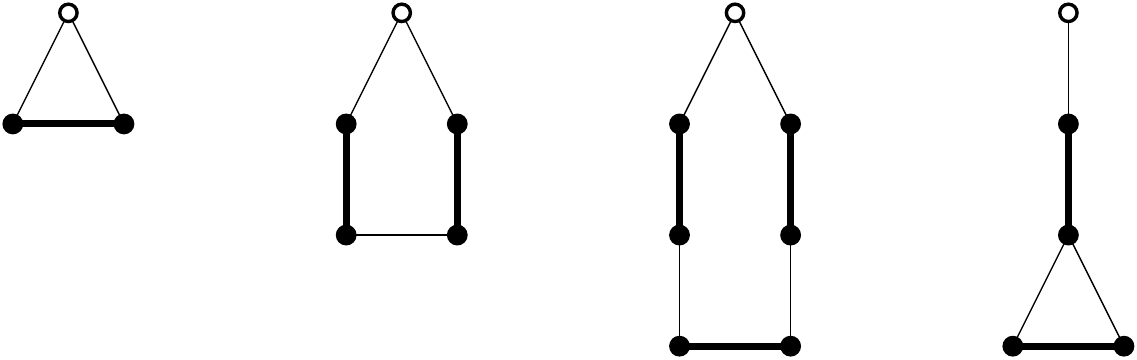}
\hspace*{\fill}
\caption{All possible flowers of length up to 7.  Free vertices
are white, matched edges are thick.}
\label{fig:flowers}
\end{figure}

Let $V_C$ 
(the ``$C$'' reminds of ``cycle'') 
be all vertices that belong to a 7-cycle-flower 
and let $M_C$ and $F_C$ be all
matching-edges and free vertices within $V_C$.

\begin{claim} 
\label{cl:FC}
\label{cl:cycleflower}
$|F_C|\leq |M_C|$.
\end{claim}
\begin{proof}
For every $f\in F_C$ there exists some 7-cycle-flower $f$-$v_1$-$v_2$-$\dots$-$v_k$-$f$
with $k\in \{2,4,6\}$.
Assign $f$ to matching-edge $(v_1,v_2)\in M_C$.  Assume for contradiction that
another vertex $f'\in F_C$ was also assigned to $(v_1,v_2)$.  Then $f'$ is adjacent to one of
$v_1,v_2$.  If it is $v_2$, then $f'$-$v_2$-$v_1$-$f$ is a 3-augmenting path.  If it is $v_1$,
then $f'$-$v_1$-$\dots$-$v_k$-$f$ is a 7-augmenting path, see Figure~\ref{fig:cycleflower}.
\end{proof}

From now on we will only study the graph $G\setminus V_C$.  
Let $F_B$ 
(the ``$B$'' reminds of ``blossom'') 
be all those free vertices $f$ that are not in $F_C$ and that belong
to a 7-flower.  By $f\not\in F_C$ this flower has a non-empty stem, 
which is possible only if its length is exactly 7 and the stem has two edges
$f$-$s$-$t$ while the blossom is a 3-cycle $t$-$x_0$-$x_1$-$t$.  Furthermore
$(s,t)$ and $(x_0,x_1)$ are matching-edges.
Let $M_B$ be the set of such matching-edges $(x_0,x_1)$ 
i.e., matching-edges that 
belong to the blossom of such a 7-flower. Note that we do {\em not} 
include the matching-edge $(s,t)$ in $M_B$ (unless it belongs to a different
7-flower where it is in the blossom).  
Let $T_B$ be the set of such vertices $t$,
i.e., vertices that belong to a 7-flower and belong to both the stem and 
the blossom.
Set $V_B=T_B\cup V(M_B)$.

\begin{claim} 
\label{cl:TDelta}
$|T_B|\leq |M_B|$.
\end{claim}
\begin{proof}
Assign each $t\in T_B$ to a matching-edge $(x_0,x_1)\in M_B$ 
that is within the same blossom of some 7-flower of $G\setminus V_C$.
Assume for contradiction that some other vertex $t'\in T_B$ is also assigned to $(x_0,x_1)$.  Let $t$-$s$-$f$ and $t'$-$s'$-$f'$ be the stems of the 7-flowers
containing $t$ and $t'$,
and note that $s\neq s'$ since they are matching-partners
of $t\neq t'$.  This gives an alternating path $f$-$s$-$t$-$x_0$-$x_1$-$t'$-$s'$-$f'$, see Figure~\ref{fig:blossom}.  Depending on whether
$f=f'$ this is a 7-augmenting path or 7-cycle-flower; the former contradicts the choice of $M$ and the latter
that $(x_0,x_1)\not\in M_C$.
\end{proof}

\begin{figure}[ht]
\hspace*{\fill}
\subfigure[~]{\includegraphics[page=14]{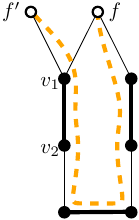}
\label{fig:cycleflower}}
\hspace*{\fill}
\subfigure[~]{\includegraphics[page=11]{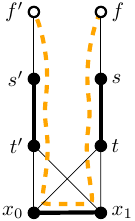}
\label{fig:blossom}}
\hspace*{\fill}
\caption{Augmenting paths found in the proofs of 
(a) Claim~\ref{cl:cycleflower} and
(b) Claim~\ref{cl:TDelta}.}
\label{fig:backedge}
\end{figure}

\paragraph{The auxiliary graph $H$}
For any vertex $v$, let the {\em distance to a free vertex} be the number of
edges in a shortest alternating path from a free vertex to $v$.  Let $D_k$
be the vertices of distance $k$ to a free vertex, and observe that
there are no edges from $D_0$ to $D_3$, else there would be
a shorter alternating path.   We consider these sets as they are found in the
graph $G\setminus V_C\setminus V_B$.

\begin{observation}
\label{obs:horizontal}
In graph $G\setminus V_C\setminus V_B$, there are no matching-edges within $D_k$ for $k = 1$ and $ k = 3 $, and no edges at all within $D_k$ for $k = 0$ and $ k = 2 $.
\end{observation}
\begin{proof}
If there was such an edge $(v,v')$, then it, together with the alternating paths of length $k$ that 
lead from free vertices to $v,v'$, form a 7-augmenting path or a 7-flower.
\end{proof}

From now on, we will only study the subgraph $H$ 
induced by $D_0\cup \dots \cup D_3$, noting again that this does
{\em not} include the vertices in $V_C\cup V_B$.  
For ease of referring to them, we rename the vertices of $H$ as follows
(see also Figure~\ref{fig:setup}):
\begin{itemize}
\item $F_H=F\setminus F_C=D_0$ are the free vertices in $H$.
\item $S=D_1$ are the vertices adjacent to free vertices in $H$.
\item $T_H=D_2$ are the vertices in $H$
	that have matching-partners in $S$ and are not in $S$.
\item $U=D_3$ are the vertices in $H$
	that are adjacent to $T_H$ and not in $F\cup S\cup T_H$.
\end{itemize}

\begin{figure}[ht]
	\centering
	\includegraphics[page=1]{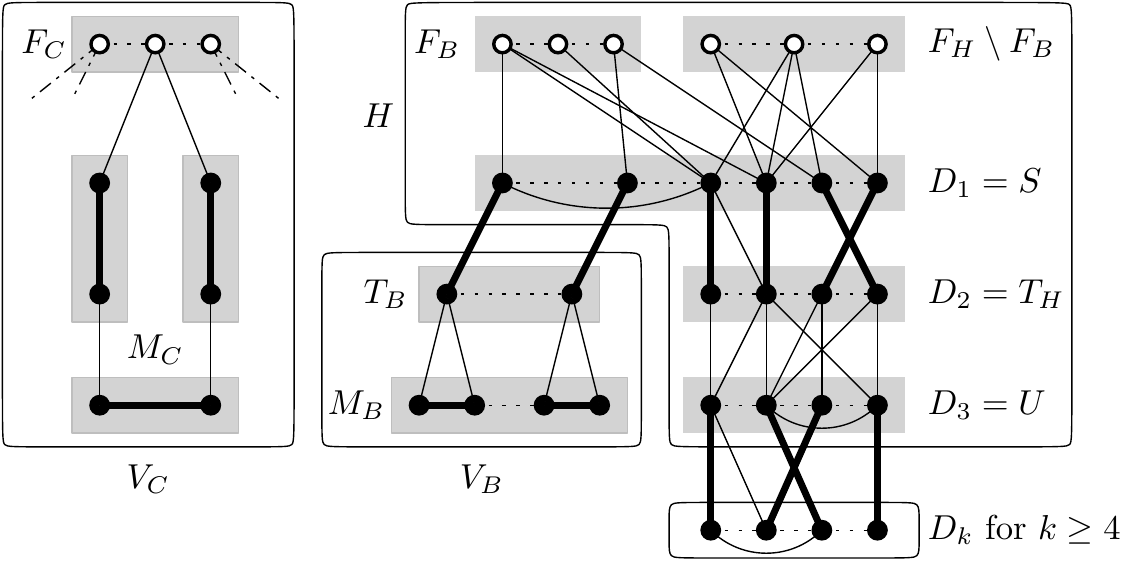}
	\caption{Illustration of the partitioning of edges and vertices. 
		The free vertices $F = F_C \cup F_H$, the matching edges $M = M_C \cup M_S \cup M_U$, and the remaining vertices $S\cup T_H\cup U$ in $H$.}
	\label{fig:setup}
\end{figure}

The following shortcuts will be convenient.  For any vertex sets $A,B$,
an {\em $A$-vertex}
is a vertex in $A$, an {\em $A$-neighbour} is a neighbour of a vertex in $A$,
and an {\em $AB$-edge} is an edge connecting a vertex in $A$ with a vertex in $B$.  
Using Observation~\ref{obs:horizontal} and the definition of $V_C$
(which includes the {\em entire} flower) and $V_B$ (which includes
{\em both} ends of the matching-edge) one easily verifies the following:

\begin{observation}
\label{obs:Hbipartite}
\begin{itemize}
\item There are no matching-edges within $S$ or within $U$.
\item There are no edges within $F_H$ or within $T_H$.
\item The matching-partner of an $S$-vertex is in $T_H\cup T_B$.
\item The matching-partner of a $U$-vertex is not in $H$.
\item All neighbours of an $F_H$-vertex belong to $S$ or are not in $H$.
\item All neighbours of a $T_H$-vertex belong to $S\cup U$ or are not in $H$.
\end{itemize}
\end{observation}

Let $M_S$ be the set of matching-edges incident to $S$.
Let $M_U$ be the matching-edges incident to $U$.
Since there are no matching-edges within $S$ or $U$, we have
$|M_S|=|S|$ and $|U|=|M_U|$.

We stated earlier that any neighbour of $F_H$ is either in $S$ or not in
$H$.  The latter is actually impossible (though this is non-trivial), and
likewise for $T_H$.

\begin{figure}
	\hspace*{\fill}
	\subfigure[~]{\includegraphics[page=2]{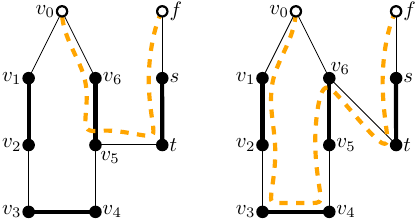}
	\label{fig:backedge_cycle}}
	\hspace*{\fill}
	\subfigure[~]{\includegraphics[page=3]{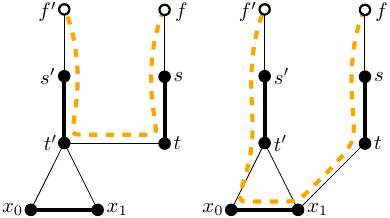}
	\label{fig:backedge_blossom}}
	\caption{Augmenting paths found in the proofs of
		(a) Lemma~\ref{lem:no_back_edge}, $t\in T_H$ has a neighbour in $V_C$.
		(b) Lemma~\ref{lem:no_back_edge}, $t\in T_H$ has a neighbour in $V_B$.}
\end{figure}

\begin{lemma}
\label{lem:no_back_edge}
No vertex in $F_H\cup T_H$ has a neighbour in $G$ that is outside $H$.
\end{lemma}
\begin{proof}
First observe that no edge can connect a vertex in $F_H\cup T_H=D_0\cup D_2$ with a vertex $z\in D_k$ for $k\geq 4$ since $z$ would have been added to $D_1=S$ or $D_3=U$ instead.  So we must only show that no vertex in $F_H\cup T_H$ has a neighbour in $V_C\cup V_B$.
We show the claim only for $t\in T_H$; the proof is similar (and even easier) for $f\in F_H$ by replacing the path $t$-$s$-$f$ defined below with just $f$.  

Figure~\ref{fig:backedge_cycle} illustrates the following.
Fix some $t\in T_H$, let $s\in S$ be its matching-partner and let $f\in F_H$ be an arbitrary free vertex incident to $s$.  
Assume for contradiction that $t$ has a neighbour $v_i$ in $V_C$, so $v_i$ belongs to some
7-cycle-flower $v_0$-$v_1$-$\dots$-$v_k$-$v_0$ where $k\in \{2,4,6\}$ and $v_0\in F$.  Note that $v_0\neq f$
since $v_0\in F_C$ while $f\in F_H$.  If $i$ is odd then 
$f$-$s$-$t$-$v_i$-$\dots$-$v_k$-$v_0$ is a 9-augmenting path, and if $i$ is even then
$f$-$s$-$t$-$v_i$-$v_{i-1}$-$\dots$-$v_1$-$v_0$ is a 9-augmenting path; both are impossible.

Now consider some $(x_0,x_1)\in M_B$ that belongs to a 7-flower $f'$-$s'$-$t'$-$x_0$-$x_1$-$t'$-$s'$-$f'$ where $(s',t')$ is a matching-edge and $t'\in T_B$.  
Note that $t'\neq t$ (hence $s'\neq s$) since $t'\in T_B$ while $t\in T_H$.  
If $t$ and $t'$ are adjacent, then $f$-$s$-$t$-$t'$-$s'$-$f'$
is a 5-augmenting path or a 5-cycle-flower.  If $t$ and $x_i$ are adjacent for $i\in \{0,1\}$, then $f$-$s$-$t$-$x_i$-$x_{1-i}$-$t'$-$s'$-$f'$
is a 7-augmenting path or 7-cycle-flower.  
See Figure~\ref{fig:backedge_blossom}.
There cannot be such augmenting paths, and no such cycle-flowers either since $t\not\in T_C$.
\end{proof}

In particular, if a vertex in $F_H\cup T_H$ had degree $d$ in $G$, then it also
has degree $d$ in $H$; this will be important in obtaining matching-bounds
below.

\paragraph{Minimum degree 3}
With this, we can prove our first matching-bound.    We need the following lemma shown by
Biedl and Wittnebel: 

\begin{lemma}[\cite{BW19}]
\label{lem:BW}
Let $G$ be a simple 1-planar graph.  Let $A$ be a non-empty independent
set of $G$ where all vertices in $A$ have degree 3 or more in $G$.
Let $A_d$ be the vertices of degree $d$ in $A$.
Then 
$$2|A_3|+\sum_{d>3} (3d-6) |A_d| \leq 12|V\setminus A|-24.$$
\end{lemma}

\begin{lemma}
\label{lem:middle}
We have (i) $|F_H|\leq 6|S|-12$ and (ii) $|F_H|+|T_H|\leq 6|S|+6|U|-12$.
\end{lemma}
\begin{proof}
Consider first the subgraph of $H$ induced by $F_H$ and $S$.
By Observation~\ref{obs:Hbipartite} and Lemma~\ref{lem:no_back_edge} 
any vertex in $F_H$ has degree at least 3
in this subgraph, and they form an independent set.  Consider the inequality of Lemma~\ref{lem:BW}. 
Any vertex in $F_H$ contributes at least 2 units to the LHS 
while the RHS is $12|S|-24$. This proves Claim (i) after dividing.

Now consider the full graph $H$. 
By Observation~\ref{obs:Hbipartite} and Lemma~\ref{lem:no_back_edge} 
any vertex in $F_H\cup T_H$ has degree at least 3
in $H$, and they form an independent set. Claim (ii) now follows from
Lemma~\ref{lem:BW} as above.
\end{proof}

\begin{corollary}
If the minimum degree is 3, then $|M|\geq \frac{7}{50}(n+12)$. 
\end{corollary}
\begin{proof}
Adding Lemma~\ref{lem:middle}(ii) six times to Lemma~\ref{lem:middle}(i) gives
$$ 7|F_H| + 6|T_H| \leq 42|S|+36|U|-84 \leq  42|M_S| +36|M_U|-84.$$
Adding Claim~\ref{cl:FC} seven times and Claim \ref{cl:TDelta} six times gives
$$ 7|F_C|+ 7|F_H| + 6|T_B|+ 6|T_H| \leq 42|M_S| +36|M_U|+7|M_C|+6|M_B|-84.$$
Since $|S|=|M_S|=|T_H|+|T_B|$, this simplifies to
$$ 7|F| =7|F_H| + 7|F_C| \leq 36|M_S|+ 36|M_U|+7|M_C|+6|M_B|-84 \leq 36|M|-84.$$
Therefore
$ 2|M|=n-|F| \geq n+12-\frac{36}{7}|M| $
which gives the bound after rearranging.
\end{proof}

It is worth pointing out that this result (as well as Theorem~\ref{thm:highdegree} below)
does not use 1-planarity of the graph except when using the bound in 
Lemma~\ref{lem:BW}.  Hence, similar bounds could be proved for any graph class where
the size of independent sets can be upper-bounded relative to its minimum degree.

Doing the improvement from $\frac{7}{50}$ to $\frac{1}{7}$ will be done by improving the bound
on $|F_H|+|T_H|$ slightly.
By modifying $H$ and its 1-planar drawing and studying a resulting 1-planar bipartite graph $J$, 
we will show the following in Section~\ref{sec:appendix_hard}:

\begin{lemma}
\label{lem:hard}
\label{lem:uvertices}
$|F_H| + |T_H| \leq 6|S| + 5|U| - 12.$
\end{lemma}

This then gives our main result:

\begin{theorem}
\label{thm:main}
Let $G$ be a 1-planar graph with minimum degree 3, and let $M$ be a matching in $G$ that
has no augmenting path of length 9 or less.  Then $|M|\geq \frac{n+12}{7}$.
\end{theorem}
\begin{proof}
Using $|S|=|M_S|$ and  $|U| = |M_U|$ we have
\begin{align*}
	|F_H|+|T_H| & \leq 6 |M_S| + 5 |M_U|-12 &&\text{from Lemma~\ref{lem:uvertices}} \\
	|F_C|  & \leq  |M_C| &&\text{from Claim~\ref{cl:FC}} \\
	|T_B|  & \leq  |M_B| &&\text{from Claim~\ref{cl:TDelta}.} 
\end{align*}
Since $|T_H|+|T_B|=|M_S|$ this gives
$
|F|+|M_S| \leq |M_C| + |M_B| + 6 |M_S| + 5 |M_U|-12,
$
therefore $|F|\leq 5|M|-12$.  This implies 
$2|M| = n-|F| \geq n-5|M|+12$ or $7|M|\geq n+12$.
\end{proof}

\paragraph{Higher minimum degree}

Since the bound for independent sets in 1-planar graphs gets 
smaller when the minimum 
degree is larger, we can prove better matching-bounds for higher minimum
degree.

\begin{lemma}
\label{lem:middle4}
If the minimum degree is $\delta>3$, then 
$$(i)\; |F_H|\leq \tfrac{4}{\delta-2}(|S|-2) \quad \text{and} \quad (ii)\; |F_H|+|T_H|\leq \tfrac{4}{\delta-2}(|S|+|U|-2).$$
\end{lemma}
\begin{proof}
As in Lemma~\ref{lem:middle}, consider the subgraph of $H$ induced by $F_H$
and $S$.  Any $f\in F_H$ has degree $\delta$ or more and contributes
at least $3\delta-6$ units to the LHS of the inequality in Lemma~\ref{lem:BW}.
The RHS is $12|S|-24$. This proves Claim (i) after dividing.  Claim (ii) is proved the same way using the full graph $H$.
\end{proof}

\begin{theorem}
\label{thm:highdegree}
Let $G$ be a 1-planar graph with minimum degree $\delta$.
Let $M$ be any matching in $G$ without 9-augmenting path.  Then 
\begin{itemize}
\item $|M|\geq \frac{3}{10}(n+12)$ for $\delta=4$,
\item $|M|\geq \frac{1}{3}(n+12)$ for $\delta \geq 5$.
\end{itemize}
\end{theorem}
\begin{proof}
Set $c=\tfrac{4}{\delta-2}$, so $|F_H|\leq c(|S|-12)$ and $|F_H|+|T_H|\leq 
c(|S|+|U|-12).$  Taking the former inequality once and adding the latter
one $c$ times gives
$$ (c+1)|F_H| + c|T_H| \leq (c^2+c) |S| + c^2 |U| - (c+1)12
= (c^2+c)|M_S| + c^2|M_U| - (c+1)12.$$
Adding Claim~\ref{cl:FC} $c+1$ times and Claim~\ref{cl:TDelta} $c$ times gives
\begin{equation}
\label{eq:highdeg}
 (c{+}1)(|F_C|{+}|F_H|) + c(|T_B|{+}|T_H|) 
	\leq (c^2{+}c) |M_S| + c^2 |M_U| + (c{+}1)|M_C| + c|M_B|- (c{+}1)12.
\end{equation}
For $\delta=4$ we have $c=2$, and with $|T_B|+|T_H|= |M_S|$ 
Equation~\ref{eq:highdeg} simplifies to 
$$ 3|F| \leq 4|M_S|+ 4|M_U|+3|M_C|+2|M_B|-36 \leq 4|M|-36.$$
Therefore $ 2|M|=n-|F| \geq n+12-\frac{4}{3}|M|$.
For $\delta\geq 5$ we have $c^2<c+1$ and so can only simplify Equation~\ref{eq:highdeg} to
$$ (c+1)(|F_C|+|F_H|) \leq (c+1)|M| - (c+1)12 $$
hence $2|M|=n-|F|\geq n+12 - |M|$.
The bounds follow after rearranging. 
\end{proof}

For $\delta=4,5$ these are close to the bounds
of $\frac{1}{3}(n+4)$ (for $\delta=4$) 
and $\frac{1}{5}(2n+3)$ (for $\delta=5$) 
that we know to be the tight lower bounds on the maximum matching size \cite{BW19}.
Unfortunately we do not know how improve Theorem~\ref{thm:highdegree} for $\delta>3$; the techniques of
Section~\ref{sec:appendix_hard} do not work for higher minimum degree since we will use
another inequality (Observation~\ref{obs:Tsmall}(i)) that is not strong enough
to achieve the bound for higher degrees, and not easily improved.

The case $\delta\geq 6$ is also interesting.  Here one would hope for even
larger matching-bounds.  Unfortunately, 
the bottleneck in our analysis is our treatment of flowers of length 3.
Here we remove one free vertex and one matching-edge, which can at best lead to
a bound of $|M|\geq \frac{1}{3}(n+O(1))$.   So a further improvement of the
bound for minimum degree $ \delta\geq 6 $ would require treating short flowers differently.

\subsection{Stopping earlier?}

Currently we remove all augmenting paths up to length 9 before returning the
matching.  Naturally one wonders whether one could stop earlier?

It is possible to show that it would suffice to remove only 7-augmenting
paths.  Inspecting the analysis, one sees that the absence of augmenting paths 
of length exactly 9 is used only once: 
In the proof of Lemma~\ref{lem:no_back_edge}, we use it to argue that a vertex 
$t\in T_H$ is not adjacent to a 7-cycle flower.  Digging further,
one can verify that 7-cycle-flowers need to be removed only to avoid
matching-edges within $U$.  It turns out that one can deal with matching-edges
within $U$ directly, by arguing that at most three vertices in $T_H$ can
have an endpoint at such an edge (else there is a 7-augmenting path), and 
removing these vertices and matching-edges and accounting for them directly.  
The details are not difficult but tedious and require even more notation;
we will not give them.

On the other hand, it is not enough to remove only 3-augmenting paths.
Figure~\ref{fig:counterthreeaug} shows an example of a matching in a 1-planar
graph that has no 3-augmenting paths, but only size $\frac{n+12}{8}$.
We can show that this is as bad as it can get.

\begin{figure}
	\centering
	\includegraphics{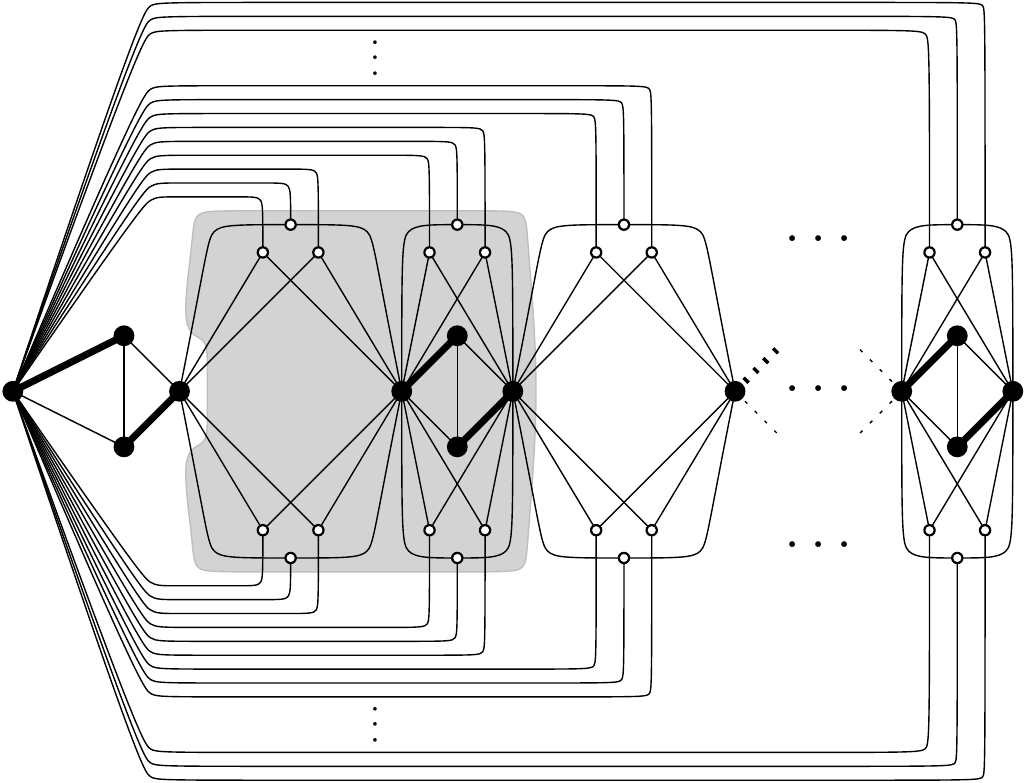}
	\caption{A graph with a matching marked in thick edges of size $ \frac{n+12}{8} $. No 3-augmenting path exists for the chosen matching, but there are 5-augmenting paths. The gray area marks an example of 16 vertices such that only 2 matching edges exist. Repeating this configuration gives the example for arbitrary $ n $.}
	\label{fig:counterthreeaug}
\end{figure}

\begin{theorem}
\label{thm:easy}
Let $G$ be a 1-planar graph with minimum degree 3 and let $M$ be a 
matching without 3-augmenting paths.  Then $|M|\geq \frac{n+12}{8}$.
\end{theorem}
\begin{proof}
The proof is very similar to the one of Theorem 2 in \cite{frankeComputingLarge2011}
except that we use Lemma~\ref{lem:BW} rather than the edge-bound
for planar bipartite graphs.  We repeat it here for completeness,
mimicking their notation.
Let $M_c$ be all those matching-edges $(x,y)$ for which some free vertex 
$f\in F$ is adjacent to both $x$ and $y$, and let $F_c$ be all such free vertices. Vertex $f$ is necessarily the only $F$-neighbour of $x$ and $y$, else there would
be a 3-augmenting path.  Hence $|F_c|\leq |M_c|$.  

Let $M_o$ and $F_o$ be the remaining matching-edges and free vertices.
For each edge $(x,y)$ in $M_o$, at most one of the ends can have $F$-neighbours,
else $(x,y)$ would be in $M_c$ or there would be a 3-augmenting path.
Let $S$ be the ends of edges in $M_o$ that have $F$-neighbours, and
let $G'$ be the auxiliary graph induced by $F$ and $S$.  Since $F$
is an independent set, we have $|F_o|\leq 6|S|-12 \leq 6|M_o|-12$
by Lemma~\ref{lem:BW}.

Putting both together, $2|M|=n-|F|\geq n+12-|M_c|-6|M_o|\geq n+12-6|M|$
and the bound follows after rearranging.
\end{proof}


\section{Proof of Lemma~\ref{lem:hard}}
\label{sec:appendix_hard}

In this section, we prove Lemma~\ref{lem:hard}, i.e., we show that
$$|F_H| + |T_H| \leq 6|S| + 5 |U| - 12.$$
\iffull
The following proof does not quite work, but puts us in the right direction.
Consider the graph $H$ defined earlier.
For any $t\in T_H$ that has $U$-neighbours, contract $t$ into one of its $U$-neighbours.
In the resulting graph all $F_H$-vertices and the remaining $T_H$-vertices have three $S$-neighbours.
If a $U$-vertex had $d$ $T_H$-vertices contracted into it, then it now has degree at least $d$
(because the matching-partners of the $T_H$-vertices are distinct).
Also, all vertices not in $S$ form an independent set.  If the resulting graph is 1-planar, we could
hence use Lemma~\ref{lem:BW} to bound $|F_H|+|T_H|$ as sufficiently small.

Unfortunately, 1-planar graphs are not closed under contraction in general;
we can only contract along uncrossed edges.  So we need to proceed more carefully,
the following is an outline.
We first assign $T_H$-vertices to
a carefully chosen $U$-neighbour ($T_H$-verties that have no $U$-neighbours
will not be contracted).  Next eliminate $U$-vertices with few assigned
$T_H$-vertices; these can be accounted for easily.  In the
resulting drawing $I$, we contract each remaining $T_H$-vertex
that has $U$-neighbours along an uncrossed edges.
Unfortunately we cannot always do a contraction along
the matching-edge or the assignment-edge, which will make it more difficult to argue
that a vertex in $U$ retain sufficiently high degrees.
Letting $J$ be the resulting final drawing,
we then apply Lemma~\ref{lem:BW} (in a stronger form) to the bipartite
graph $J$ to prove Lemma~\ref{lem:hard}.  
\fi
We phrase the various procedures below
as if they were algorithms, but remind the reader that they are only
used for the analysis and not needed for finding the matching.

\paragraph{The super-graph $H^+$.}

Fix an arbitrary 1-planar drawing of $H$;    from now on we use
$H$ (as well as the graphs $H^+,I,J,J^-$ derived from it) to mean both
the graph and the 1-planar drawing that comes with it, and which will
not be changed unless stated explicitly.

We obtain a 1-planar drawing $H^+$ by changing $H$ in two ways.
First, delete all edges of $H$ that are within $S\cup U$.
Second, add/re-route kite-edge as follows.
Let $(t,x)$ 
by a potential kite-edge of some crossing $c$ of $H$ with $t\in T_H$
and $x\in U\cup S$.
If $(t,x)$ does not yet exist in $H$, then add it.
If $(t,x)$ exists already in $H$ and is crossed, then re-route $(t,x)$ 
as the kite-edge so that it becomes uncrossed.  If $(t,x)$ exists as
uncrossed edge in $H$ already, then do not insert it (so that $H^+$
stays simple).  However, if crossing $c$ involves the matching-edge at $t$,
then re-route $(t,x)$ (if needed) to be at crossing $c$ rather than elsewhere.
Repeat for all such potential kite-edges $(t,x)$.
See also Figure~\ref{fig:example2}.

\begin{figure}[ht]
\hspace*{\fill}
\subfigure[~]{\includegraphics[page=12,width=0.45\linewidth]{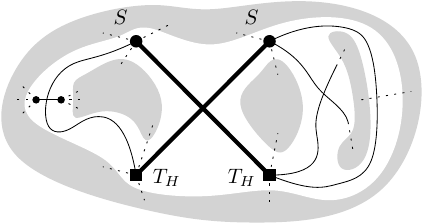}}
\hspace*{\fill}
\subfigure[~]{\includegraphics[page=13,width=0.45\linewidth]{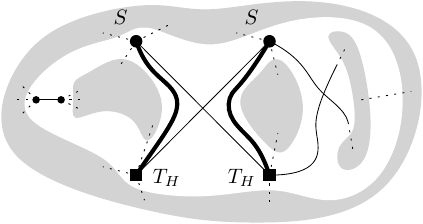}}
\hspace*{\fill}
\caption{Re-routing kite-edges that are crossed or not routed near
the matching edge, and trading matching-partners. The gray areas symbolize other parts of the graph. Dotted edges indicate possible connections.}
\label{fig:example2}
\end{figure}

We also {\em trade matching-partners} as follows.  Assume that
$H^+$ contains two matching edges $(s,t)$ and $(s',t')$ 
(with $s,s'\in S$ and $t,t'\in T_H$) that cross each other.
We inserted (or re-drew)
kite-edges $(s,t')$ and $(s',t)$ at this crossing.
We now remove edges $(s,t)$ and $(s',t')$ from the matching
and declare $(s,t')$ and $(s',t)$ to be matching-edges instead, 
noting that these are uncrossed.
This exchange would not necessarily have been possible in $G$
where these kite-edges need not exist, but we do it here only
for the purpose of explaining how to transform $T_H$-vertices
and do not actually change the returned matching.
From now on, ``matching-edge''
and ``matching-partner'' refers to the status after this trading.
Note that all vertices in $T_H$ remain matched.

We summarize the properties of $H^+$ for future reference:

\begin{observation}
\label{obs:Hplus}
$H^+$ is a simple graph with a fixed 1-planar drawing that satisfies the following.
\begin{enumerate}
\item[(a)] Any edge has exactly one endpoint in $S\cup U$.
\item[(b)] Vertices in $F_H$ have degree 3 or more and all their neighbours are in $S$.
	Vertices in $T_H$ have degree 3 or more and all their neighbours are in $S\cup U$.
\item[(c)] If $(x,t)$ is a potential kite-edge at crossing $c$ for some 
$x\in S\cup U$ and $t\in T_H$, 
then $(x,t)$ exists as an uncrossed edge in $H^+$.  If $c$ involves the matching-edge
at $t$, then $(x,t)$ is routed at crossing $c$.
\item[(d)] No two crossing matching-edges cross each other.
\end{enumerate}
\end{observation}
\begin{proof}
To show (a), observe that in $H$ every edge has at least one endpoint in $S\cup U$
by Observation~\ref{obs:Hbipartite}.  Added kite-edges also have an endpoint in $S\cup U$,
There is exactly one such endpoint since we deleted edges within $S\cup U$.
(b) held in $H$ by Observation~\ref{obs:Hbipartite} and also holds for added kite-edges.
(c) and (d) hold by construction of $H^+$
and since we traded matching-partners.  
\end{proof}

\paragraph{Assignment to $U$-vertices.}
We assign vertices in $T_H$ to $U$-neighbours with the following {\sc AssignmentAlgorithm}.
For each $t\in T_H$:
\begin{itemize}
\item If $t$ has an uncrossed edge to a $U$-neighbour:
	\begin{itemize}
	\item Add $t$ to a set $T_\mu$.    (As we will see,
		vertices in $T_\mu$ use their matching-edges during the transformation
	explained later, hence the ``$\mu$''.)
	\item If the matching-edge $(t,s)$ is crossed by an edge $(y,u)$
		with $u\in U$, then assign $t$ to this $u$ (noting that
		edge $(t,u)$ exists as uncrossed kite-edge at this crossing).
	\item Otherwise assign 	$t$ to an arbitrary $U$-neighbour $u$ for which
		edge $(t,u)$ is uncrossed.
	\end{itemize}
\item Otherwise, if $t$ has at least three $S$-neighbours, add $t$ to $T_\sigma$.
	(The ``$\sigma$'' reminds that there are sufficiently
	many $S$-neighbours that we do not need to transform these vertices.)
\item Otherwise, all edges to $U$-neighbours are crossed.
	Add $t$ to a set $T_\rho$ (``$\rho$'' stands for ``remaining'').  
%
	Assign $t$ to an arbitrary $U$-neighbour.
\end{itemize}

If $t\in T_H$ has been assigned to $u\in U$, then
we call $(t,u)$ the {\em assignment-edge} of $t$.
Let $U^d$ be the set of all those vertices $u\in U$ that had $d$ incident
assignment-edges. 
It will be helpful to assign as many $T$-vertices as possible to 
$\bigcup_{d\leq {5}} U^d$;
we therefore do the following {\em re-assignment}:  If $t\in T_H$ is
assigned to a vertex $u\not\in \bigcup_{d\leq {5}} U^d$,
but $t$ has a neighbour $u'\in U^0\cup \dots \cup U^{4}$,
then re-assign $t$ to $u'$.  Note that this changes the sets $U^d$, but it
can only increase $\bigcup_{d\leq {5}} U^d$ since we never assign another
vertex of $T_H$ to a vertex in $U^5$.
Repeat (with the new sets $U^0,\dots,U^{5}$) until no more re-assignments are possible.

Let $T^d$ be the vertices in $T_H\setminus T_\sigma$ that are assigned to
a vertex in $U^d$.
So $|T^d|=d|U^d|$, and with our choice of assignment-edge, we have
the following trivial (but crucial) properties:

\begin{observation}
\label{obs:Tsmall}
\begin{enumerate}
\item[(i)] $|T_0|=0$ and $\sum_{d=1}^{5} |T^d|\leq 5\sum_{d=0}^{5} |U^d|.$
\item[(ii)] If $t\in T^d$ with $d\geq 6$ has an uncrossed edge to a $U$-neighbour, then $t$ belongs to $T_\mu$.
\item[(iii)] If $t\in T^d$ with $d\geq 6$ has a neighbour $u$  in $\bigcup_{j\leq 5} U^j$, then $u\in U^5$.
\end{enumerate}
\end{observation}

In particular, using (i) the vertices in $\bigcup_{d\leq 5} T^d$ can be accounted for
directly by the matching edges  incident to $\bigcup_{d\leq 5} U^d$.
We only need to bound the remaining $T_H$-vertices, which are those in $T_\sigma$ 
as well as those in $\bigcup_{d\geq 6} T^d$.    

\paragraph{The easy transformations.}

We now explain the contractions mentioned in the outline.  However, we prefer to view
them as ``transformations'' (consisting of deleting one vertex $t$ and inserting
one edge $(s,u)$ between neighbours of $t$) since the route used for the new edge must
be described carefully.  We first do the ``easy'' transformations where $s$ can be chosen
to be the matching-partner of $t$.  Let $I$ be the 1-planar drawing
obtained from $H^+$ as follows:

\begin{itemize}
\item Delete all vertices in $\bigcup_{d\leq 5} (U^d \cup T^d)$.
\item For any remaining vertex $t\in T_H$, delete all edges to $U$-neighbours except
	the assignment-edge (if any).
\item While there exists a vertex $t\in T_H\setminus T_\sigma$ for which either the
		assignment-edge $(t,u)$ or the matching-edge $(s,t)$ is uncrossed:
	\begin{itemize}
	\item Delete $t$.
	\item If the path $s$-$t$-$u$ had a crossing $c$ with an edge $(x,y)$ that
		has an endpoint in $\{s,u\}$:  Insert $(s,u)$ as a kite-edge of crossing $c$.
		We call this a {\em $\kappa$-transformation} (the ``$\kappa$'' reminds of ``kite'').
	\item Otherwise insert $(s,u)$, drawing it along path $s$-$t$-$u$, and observe that
		this has at most one crossing and gives a good drawing.  We call this
		a {\em $\pi$-transformation} (the ``$\pi$'' reminds of ``path'').
	\end{itemize}
\end{itemize}

Figure~\ref{fig:example}(a-c) shows an example of this transformation algorithm.
Note that the transformation will remove {\em all} vertices in $T_\mu$, and it may also
remove some vertices of $T_\rho$, since their matching-edges and/or assignment-edges may
become uncrossed as other vertices are deleted.  Let $T_\mu'$ be all those vertices
that were removed and let $T_\rho'\subseteq T_\rho$ be the vertices $T_H\setminus \bigcup_{d\leq 5}T^d \setminus T_\sigma\setminus T_\mu'$,
i.e., all those vertices that still need to be removed.

\renewcommand{\floatpagefraction}{0.9}
\begin{figure}[ht]
\subfigure[~]{\includegraphics[page=4,width=0.44\linewidth]{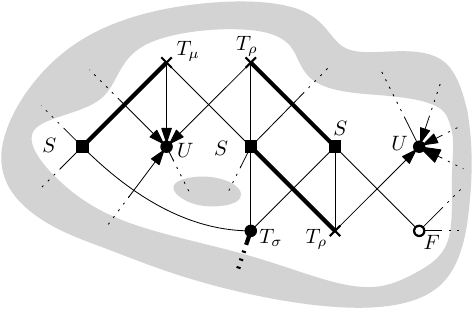}}
\hspace*{\fill}
\subfigure[~]{\includegraphics[page=5,width=0.44\linewidth]{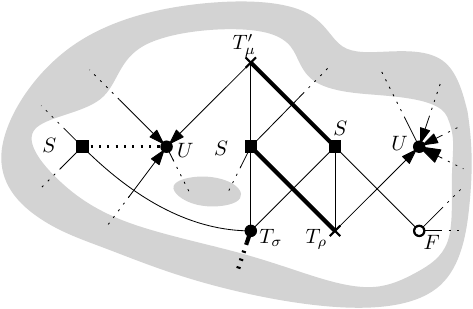}}
\newline
\subfigure[~]{\includegraphics[page=6,width=0.44\linewidth]{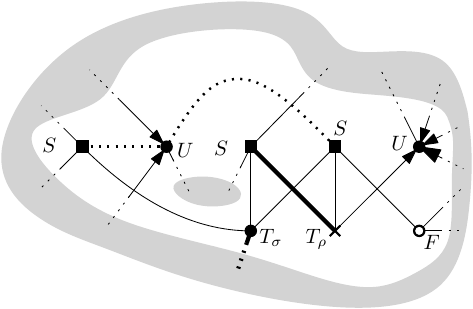}}
\hspace*{\fill}
\subfigure[~]{\includegraphics[page=7,width=0.44\linewidth]{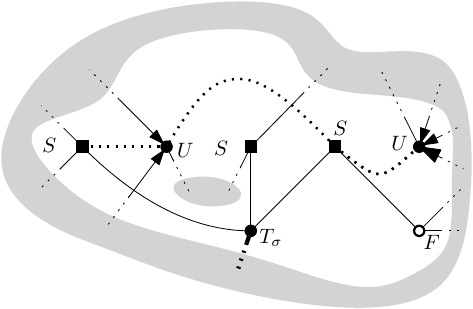}}
\caption{(a) A part of a graph $H^+$;  arrows indicate assignments.  Vertices are labeled with the sets they belong to . Gray areas symbolize the remaining graph and its incident dotted edges are possible connections.  (b) An $\kappa$-transformation:  The new edge gets routed as kite-edge.  This causes one vertex in $T_\rho$ to move to $T_\mu'$.  (c) A $\pi$-transformation:  The new edge gets routed along the path through the eliminated vertex.  (d) A $\rho$-transformation.}
\label{fig:example}
\label{fig:transform}
\end{figure}

\begin{observation}
Drawing $I$ is simple and good and any vertex $u\in U^d$ with $d\geq 6$ has degree $d$ in $I$.
\end{observation}
\begin{proof}
Vertex $u$ had $d$ incident assignment-edges, say to $t_1,\dots,t_d$.  For each $t_i$,
if $t_i\in T_\rho'$ then it remains a neighbour in $I$.  If $t_i\in T_\mu'$, then
edge $(t_i,u)$ in $H^+$ has been replaced by edge $(s_i,u)$ in $I$, where $s_i$ is
the matching-partner of $t_i$.  This does not create multiple edges since we deleted
all $SU$-edges that existed previously in $H^+$, and since no $s_i$ is matched to
two vertices.  The drawing is good since we use an $\kappa$-transformation if drawing
along the path $s_i$-$t_i$-$u$ would have violated goodness.  
\end{proof}

We call a region a {\em kite-region} if it is bounded by two half-edges at
a crossing $c$ and an uncrossed edge (necessarily a kite-edge of $c$).

\begin{figure}
\hspace*{\fill}
\subfigure[~]{\includegraphics[page=6]{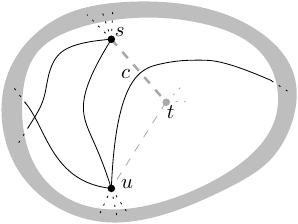}}
\hspace*{\fill}
\subfigure[~]{\includegraphics[page=5]{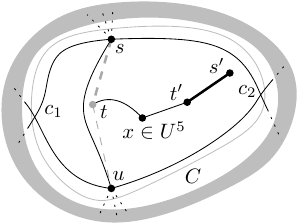}}
\hspace*{\fill}
	\caption{Illustration of the situation in Claim~\ref{cl:kite_region} if $ t $ underwent (a) a $\kappa$-transoformation, and (b) a $ \pi $-transformation and $ x \in U^5 $.}
	\label{fig:notwokiteregions}
\end{figure}
\begin{claim}
\label{cl:kite_region}
Let $(s,u)$ be an uncrossed edge in $I$ with $s\in S$ and $u\in U$.  Then
at most one of the two regions incident to $(s,u)$ in $I$ is a kite-region. 
\end{claim}
\begin{proof}
Since we deleted $SU$-edges of $H^+$ as first step of the transformation,
edge $(s,u)$ was added during the transformation, say at vertex $t\in T_\mu'$
for which $(s,t)$ was a matching-edge.
If this was an $\kappa$-transformation, then $(s,u)$ is drawn as kite-edge
of some crossing $c$ that involved either $(s,t)$ or $(t,u)$.  But the edges
$s$-$t$-$u$ are removed from the drawing, so $c$ no longer exists as crossing
in $I$.  So in $I$, the region incident to $(s,u)$ where $c$ used to be is not a kite-region.
See Figure~\ref{fig:notwokiteregions}(a).

Now assume that $t$ underwent a $\pi$-transformation, and assume for contradiction
that there are two crossings $c_1,c_2$ at which the inserted edge $(s,u)$ bounds 
kite-regions.   Therefore in drawing $I$ we have the curve $C:=s$-$c_1$-$u$-$c_2$-$s$
that contains $(s,u)$ on one side (say the inside) and everything else on the
other side. Figure~\ref{fig:notwokiteregions}(b) gives an illustration. 
In drawing $H^+$, curve $C$ also existed, and contained $t$ on the
inside.  Since no crossed edge can be crossed again,
any neighbours of $t$ in $H^+$ lies on or inside $C$.  There are at least three
neighbours of $t$ by Observation~\ref{obs:Hplus}.  Two of these are $s$ and $u$,
but there exists a third neighbour $x\in S\cup U$ of $t$ inside $C$.  If $x\in S$ or
$x\in U^d$ for $d\geq 6$ then
$x$ also exists in $I$ and lies inside $C$, contradicting that we had only the
kite-regions inside $C$ in $I$.  So $x\in \bigcup_{d\leq 5} U^d$, which by Observation~\ref{obs:Tsmall}
means $x\in U^5$.  So at least one other $T_H$-vertex $t'$ is assigned to $x$.  Let $s'$
be the matching-partner of $t'$, so $s'\neq s$.  Path $x$-$t'$-$s'$ is disjoint from $s,u$ and therefore
must also lie inside $C$ in $H^+$.  Since $s'\in S$ remains in $I$, this is again a
contradiction.
\end{proof}

\paragraph{The vertices in $T_\rho'$.}

The vertices in $T_\rho'$ that could not be removed yet have a special
structure in their neighbourhood that is illustrated in Figure~\ref{fig:Tpsi}(a).

\begin{figure}[ht]
\hspace*{\fill}
\subfigure[~]{\includegraphics[page=1,width=0.3\linewidth]{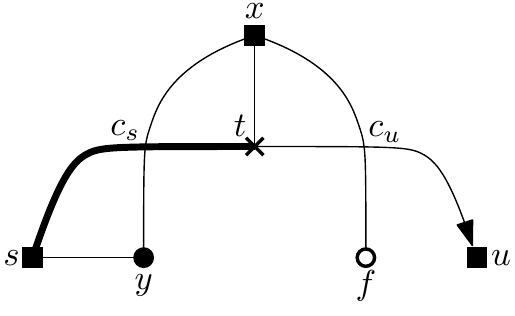}}
\hspace*{\fill}
\subfigure[~]{\includegraphics[page=2,width=0.3\linewidth]{transform_example.pdf}}
\hspace*{\fill}
\subfigure[~]{\includegraphics[page=3,width=0.3\linewidth]{transform_example.pdf}\label{fig:ifmu}}
\hspace*{\fill}
\caption{(a) The structure at a vertex $t\in T_\rho'$.  (b) Doing a $\rho$-transformation at $t$.
(c) If $y\in T_\mu$ or $f\in T_\mu$, then $c_s$ respectively $c_u$ would not exist in $I$.}
\label{fig:Tpsi}
\end{figure}

\begin{lemma}
\label{lem:Tpsi}
Fix a vertex $t\in T_\rho'$.  Let $(t,s)$ and $(t,u)$ be its matching-edge and assignment-edge.  
Then in drawing $I$:
\begin{enumerate}
\item[(a)] The matching-edge $(t,s)$ is crossed by some edge $(x,y)$ with $x\in S$.
\item[(b)] The assignment-edge $(t,u)$ is crossed by some $(x',f)$ with $x'\in S$. 
\item[(c)] $x=x'$ and the edge $(t,x)$ exists as uncrossed edge. 
\item[(d)] $y\in F_H \cup T_\sigma$.
\item[(e)] $f\in F_H$.
\end{enumerate}
\end{lemma}
\begin{proof}
We know that $(t,s)$ and $(t,u)$ are both
crossed in $I$ (else $t$ would have been transformed).  Since we never
add crossings, they were also crossed in $H^+$,
say $(t,s)$ was crossed by some edge $(x,y)$ at $c_s$ and $(t,u)$ was crossed 
by some $(x',f)$ at $c_u$.    We first show that (a-e) hold in $H^+$, i.e., 
for these vertices $x,x',y,f$.  By Observation~\ref{obs:Hplus}(a), 
up to renaming $x,x'\in S\cup U$.  By Observation~\ref{obs:Hplus}(c), 
edges $(t,x)$ and $(t,x')$ exist as uncrossed edges.  This implies
$x,x'\in S$ because
otherwise $t$ would have an uncrossed edge to a $U$-neighbour, contradicting
its crossed assignment-edge $(t,u)$.   This proves (a) and (b).

Since edge $(t,s)$ is crossed while $(t,x)$ and $(t,x')$ are uncrossed, we know $x\neq s\neq x'$.  
If $x\neq x'$ then $t$ has three $S$-neighbours and would belong to $T_\sigma$ rather than
$T_\rho$.  So $x=x'$, which proves (c).

To prove (d), recall that by Observation~\ref{obs:Hplus}(a)
and $x\in S$ we know $y\in T_H\cup F_H$.
Assume that $y\in T_H$ (else (d) holds).  Then $y$ has the $S$-neighbours
$s$ and $x$ since we added kite-edges.  Its matching-partner is not $s$ (since $(s,t)$ is
a matching-edge) and also not $x$ (else $(s,t)$ and $(x,y)$ would cross each other).
So its matching-partner is a third $S$-neighbour $s_y$.   Since $y$ has three $S$-neighbours,
we hence either add it to $T_\mu$ or to $T_\sigma$.  But we cannot have $y\in T_\mu$, else
it would have been transformed and replaced by edge $(u_y,s_y)$ for some $u_y\in U$, see Figure~\ref{fig:ifmu}.  By $s_y\neq x$ this
eliminates crossing $c_s$,
contradicting that $(t,s)$ remains crossed in $I$.  So (d) holds.

Proving (e) is similar.  By $x\in S$ we know $f\in T_H\cup F_H$.  
Assume for contradiction that $f\in T_H$.
This implies $f\in T_\mu$ since $(f,u)$ is a potential kite-edge, hence exists uncrossed in $H^+$.
So the transformations
delete $f$ and insert edge $(u_f,s_f)$ where $s_f$ and $u_f$ are the
matching-partner and assigned $U$-neighbour.  If $s_f\neq x$, 
then this will eliminate crossing $c_u$.  If $s_f=x$, then our choice of
assignment-edges ensured that $u_f=u$ and
the $\kappa$-transformation inserts $(s_f,u_f)=(x,u)$ as kite-edge of $c_u$.
See Figure~\ref{fig:ifmu}.
So either way $(t,u)$ is no longer crossed in $I$, a contradiction and (e) holds.

So we have proved (a-e) for the edges that crossed $(t,s)$ and $(t,u)$ in $H^+$.
But since their endpoints belong to $S\cup U\cup F_H\cup T_\sigma$, none of
the transformations affect these edges, so the same situation holds in $I$.
\end{proof}

We now finish our transformations as follows:
\begin{itemize}
\item For any $t\in T_\rho'$, let the assignment-edge be $(t,u)$, and let $(x,f)$
	be the edge that crosses it at $c_u$.
	Delete $t$ and insert $(x,u)$ as kite-edge of $c_u$.
	Insert this edge even if it existed already in $I$.  We call this
	a {\em $\rho$-transformation}.
\end{itemize}

Observe that the only edge that is made uncrossed by a $\rho$-transformation is
$(x,f)$.  Neither of its endpoints belongs to $T_\rho'$, so no $\rho$-transformation
affects the status of crossed edges at some other vertex $t'\in T_\rho'$.  Consequently,
Lemma~\ref{lem:Tpsi} continues to hold for all drawings during all $\rho$-transformations.
Let $J$ be the final drawing. Observe that $J$ is
specfically allowed to have multi-edges (but no loops).  It may even have an
{\em empty lens} (two uncrossed copies of an edge that bound a region).  But it
does not have an {\em empty $\Theta$}, i.e., three uncrossed copies of an edge such that
two of the three areas between them are regions of the drawing.  

\begin{claim}
If $J$ has multiple copies of some edge $(x,u)$, then all but at most one copy
are uncrossed, and $J$ has no empty $\Theta$.
\end{claim}
\begin{proof}
Since $I$ is simple, a second copy of $(x,u)$ can be inserted only during
a $\rho$-transformation.  This adds $(x,u)$ as an uncrossed kite-edge, so
only the copy of $(x,u)$ that may have existed in $I$ can be crossed.

Now assume for contradiction that $(x,u)$ exists in three uncrossed copies $e_1,e_2,e_3$ that 
bound an empty $\Theta$, say the regions bounded by $e_1,e_2$ and $e_2,e_3$ are
empty lenses.  We necessarily have $x\in S$ and $u\in U$, since no
other edges are inserted in $J\setminus I$.  
Since $I$ is simple, at least two 
of $e_1,e_2,e_3$ were inserted with $\rho$-transformations.  

A copy of $(x,u)$ is inserted during a $\rho$-transformation only if
it was the kite-edge of some crossing $c_u$ between edge $(x,f)$
(for some $f\in F_H$)
and edge $(t,u)$ (where $t\in T_\rho'$ is being transformed).
Edge $(x,f)$ becomes consecutive with the newly inserted copy of $(x,u)$,
and both are uncrossed, which means that one region incident
to this copy of $(x,u)$ is not an empty lens in $J$.  Therefore $e_2$ cannot
have been inserted with $\rho$-transformations.

The only remaining possibility is that $e_2$
already existed in $I$ while $e_1$ and $e_3$ were inserted with $\rho$-transformations.
But then both regions incident to $e_2$ in $I$ must have been kite-regions,
contradicting Claim~\ref{cl:kite_region}.
\end{proof}

\paragraph{Crossing-weighted degrees.}
For any vertex $v$, let the {\em crossing-weighted degree} of $v$ be
the degree plus the number of incident uncrossed edges.  Let $J^-$
be the drawing obtained from $J$ by deleting one edge from each empty lens.

\begin{observation}
\label{obs:degreeU}
For any vertex $u\in U^d$ with $d\geq 6$, the crossing-weighted degree of $u$ in $J^-$
is at least $d$ and the degree is at least 3.
\end{observation}
\begin{proof}
We know that $u$ had degree $d$ in $I$, and with the same argument
it has degree $d$ in $J$.  Assume $u$ was incident to $k$ empty lenses
in $J$.    Since there are no empty $\Theta$'s, edges in empty lenses
come in pairs; in particular $k\leq d/2$ and we delete $k$ edges.
Then the degree of $u$ in $J^-$ is $d-k\geq d/2=3$. The $k$ edges that
remain from the empty lenses are all uncrossed,  so $u$ has 
crossing-weighted degree at least $(d-k)+k=d$.
\end{proof}

The matching bound now follows by applying Lemma~\ref{lem:BW} to drawing $J^-$
and combining it with all other inequalities we derived earlier.    We actually
need a slightly modified version of Lemma~\ref{lem:BW}.

\begin{lemma}[\cite{BW19}]
\label{lem:BWstronger}
Let $G$ be a 1-planar graph and $A$  be a non-empty independent set in $G$
where all vertices of $A$ have degree 3 or more.  Let 
$W_d$ be the vertices in $A$ that have crossing-weighted degree $d$. 
Assume that $G$ is either simple or it has no empty lens and for any multiple edge at most
one copy is uncrossed.
Then $$2|W_3|+2|W_4| +\sum_{d\geq 5} (3d{-}12)|W_d| \leq 12|V\setminus A|-24.$$
\end{lemma}

(The lemma is stated in \cite{BW19} only for simple graphs, but as pointed
out in a later section of \cite{BW19} it holds for non-simple graphs 
under the above assumption.)

The proof of Lemma~\ref{lem:hard} is now finished as follows.
The vertices in $F_H\cup T_\sigma\cup \bigcup_{d\geq 6} U^d$ are independent in $J^-$
because all neighbours of $F_H$ are in $S$ and we deleted all $UT_\sigma$-edges during the transformations.
They also have degree 3 or more by choice of $T_\sigma$ and Observation~\ref{obs:degreeU}.
The remaining vertices of $J^-$ are $S$.
Using Lemma~\ref{lem:BWstronger} we get
\begin{eqnarray*}
12|S|-24 & \geq & 2|F_H|+2|T_\sigma| +  \sum_{d\geq 6} (3d{-}12) |U^d| 
\geq 2|F_H|+2|T_\sigma| + \sum_{d\geq 6} (2d{-}10) |U^d| \\
& \geq & 2|F_H| + 2|T_\sigma| + 2\sum_{d\geq 6} |T^d| - 10\sum_{d\geq 6} |U^d|
\end{eqnarray*}
Rearranging and adding $10 \sum_{d\leq 5} |U^d| \geq 2\sum_{d\leq 5} |T^d|$ gives
\begin{eqnarray*}
12|S|+ 10 \sum_{d\geq 0} |U^d| -24 & \geq & 2|F_H|+2|T_\sigma| + 2\sum_{d\geq 1} |T^d| = 2|F_H|+2|T_H|.
\end{eqnarray*}
This proves Lemma~\ref{lem:hard} by $|U|=\sum_d |U^d|$.

\section{Summary and outlook}
\label{sec:conclusion}

In this paper, we considered how to find a large matching in a 1-planar graph with minimum degree 3.  We argued that any matching without augmenting paths of length up to 9 has size at least $\frac{n+12}{7}$, which is the largest that one can hope for in a 1-planar graph with minimum degree 3.  Such a matching can easily be found in linear time, even if no 1-planar drawings is known, by stopping the matching algorithm by Micali and Vazirani after a constant number of rounds.

It remains open how to find large matchings in 1-planar graphs with minimum degree $\delta>3$; we can argue some lower bounds on the size of matchings without 9-augmenting paths, but these are not tight.  It would also be interesting to study other near-planar graph classes such as $k$-planar graphs (for $k>1$); here we do not even know what tight matching-bounds exist and much less how to find matchings of that size in linear time.

\bibliographystyle{plain}
\bibliography{paper.bib}
\end{document}